\documentclass[letterpaper,11pt]{article}

\usepackage{fullpage}
\usepackage{setspace}
\usepackage{amsthm}
\usepackage{amsmath}
\usepackage{amssymb}
\usepackage{algorithmic}
\usepackage{graphicx}
\usepackage{xspace}

\newtheorem{lemma}{Lemma}[section]
\newtheorem{definition}{Definition}[section]
\newtheorem{theorem}{Theorem}[section]

\newtheorem{corollary}{Corollary}[section]

\newcommand{\QBCb}{PivotBiCluster\xspace}

\newcommand{\BCC}{BCC\xspace}
\newcommand{\CC}{CC\xspace}
\newcommand{\E}[1]{\mathbb{E}\left[ #1 \right]}
\newcommand{\bT}{\bar{T}}
\DeclareMathOperator{\cost}{cost}

\def\E{\mathbb{E}}

\title{An Improved Algorithm for Bipartite Correlation Clustering} %
\author{Nir Ailon\thanks{Technion Israel Institute of Technology.} \;
            Noa Avigdor-Elgrabli\thanks{Technion Israel Institute of Technology and Yahoo! Research.} \;\;%
	    Edo Liberty\thanks{Yahoo! Research.} \;%
           }

\date{\nonumber}

\begin{document} %
\begin{spacing}{1.135}

\maketitle %
\setcounter{page}{0}
\thispagestyle{empty}
\begin{abstract}
Bipartite Correlation clustering is the problem of generating a set of disjoint bi-cliques on a set
of nodes while minimizing the symmetric difference to a bipartite input graph.
The number or size of the output clusters is not constrained in any way.

The best known approximation algorithm for this problem gives a factor of $11$.\footnote{A previously claimed $4$-
approximation algorithm \cite{GHKZ08} is erroneous, as we show in the appendix.	}
This result and all previous ones involve solving large linear or semi-definite programs
which become prohibitive even for modestly sized tasks.
In this paper we present an improved factor $4$ approximation algorithm to this problem using a simple combinatorial
algorithm which does not require solving large convex programs.

The analysis extends a method developed by Ailon, Charikar and Alantha in 2008, where a randomized pivoting
algorithm was analyzed for obtaining a $3$-approximation algorithm for Correlation Clustering, which is
the same problem on graphs (not bipartite).  The analysis for Correlation Clustering there required
defining events for  structures containing $3$ vertices and using the probability of these events to produce a feasible
solution to a dual of a certain natural LP bounding the optimal cost.

It is tempting here to use sets of $4$ vertices, which are the smallest structures for which contradictions arise for Bipartite Correlation Clustering.
This simple idea, however, appears to be evasive.
We show that, by modifying the LP, we can analyze algorithms which take into consideration subgraph structures of unbounded size.
We believe our techniques are interesting in their own right, and may be used for other problems as well.
\end{abstract}
\newpage

\section{Introduction}

Bipartite Correlation Clustering (\BCC) is a problem in which the input is a  bipartite graph and the output is a
set of disjoint clusters covering the graph nodes.\footnote{Here we consider the unweighted case, although a weighted version can be easily obtained from our analysis.}
A cluster may contain nodes from either side of the graph, but it may
also contain nodes from only one side.
We think of a cluster as a bi-clique connecting all the elements from its left and right counterparts.
An output clustering is hence a union of bi-cliques covering the input node set.
The cost of the solution is the symmetric difference between the input and the output.  Equivalently, any pair of vertices, one on the left and one of the right,
will incur a unit cost if either (1) an edge connects  them but the output clustering separates them in distinct clusters, or (2) no edge connects them but the
output clustering puts them in the same cluster.
The objective is to minimize this cost.

This notion of clustering is natural when the number of clusters and their size
are not known, and the graph relations are bipartite by nature.
It was studied in context of molecular biology, specifically, in gene expression data analysis (for example \cite{Madeira04,Cheng2000}).
Other examples for bipartite data abound.
In collaborative filtering and recommender systems interactions are given between users and items \cite{Symeonidis06}, for example, raters vs. movies/songs.
Other examples may include images vs. user generated tags and search engine queries vs. search results.

\BCC is a bipartite version of the more well known
Correlation Clustering (\CC), introduced by Bansal, Blum and Chawla \cite{BBC04},
where the objective is to cover an input set of nodes with
disjoint cliques (clusters) minimizing the symmetric difference with a given edge set over these nodes.
One motivation for \BCC, which also applies to our setting, is a \emph{2-stage} clustering
approach in which one (i) applies binary classification machine-learning methods to predict pairs of nodes that should be
clustered together, and (ii) uses the learned classifier, applied to all pairs, as input to \BCC.
Assuming  there is a correct clustering of the data and that the above binary classifier has some bounded error rate with respect to that ground truth,
we can recover, using an algorithm for \CC (or, \BCC in our bipartite case) a clustering of the data which is provably close to the true clustering (see \cite{AL09clustering}). 

Another motivation is the alleviation of the need to specify the number of output clusters, as
often needed in clustering settings such as $k$-means or $k$-median.  The treatment of clustering problems
as \CC or \BCC  should  be compared to their  predating  (by decades) statistical theory of record linkage  where, in a typical application,
one wishes to identify duplicate records in a database riddled with human errors.
The number of clusters is clearly unknown.  In fact, the original  record linkage literature \cite{RecordLinkage} considered the bipartite
case, a typical example being two government agencies cross-validating large databases of population information.

Bansal et. al \cite{BBC04}  gave a $c \approx 10^4$
factor for approximating \CC running in time $O(n^2)$ where $n$ is the number of nodes in the graph.
Later, Demaine et. al \cite{DEFI06} gave a $O(\log(n))$ approximation algorithm for
an \emph{incomplete} version of \CC, relying on solving an LP and rounding its solution by employing a region growing procedure.
By incomplete we mean that only a subset of the node pairs participate in the symmetric difference cost
calculation.\footnote{In some of the literature, \CC refers to the much harder incomplete version,
and ``\CC in complete graphs'' is used for the version we have described here.}
\BCC is, in fact,  a special case of incomplete \CC, in which the non-participating node pairs lie on the same side of the graph.
Charikar et. al \cite{CharikarGW05} provide a $4$-approximation algorithm for \CC, and another $O(\log n)$-approximation algorithm for the incomplete case.
Later, Ailon et. al \cite{ACN08} provided a $2.5$-approximation algorithm for \CC based on rounding an LP.
They also provide a simpler $3$-approximation algorithm, QuickCluster, which runs in time linear in the number of edges of the graph.
In \cite{AL09clustering} it was argued that QuickCluster runs in  expected time  $O(n + cost(OPT))$.

 Van Zuylen et. al \cite{ZHJW07} provided de-randomization
for the algorithms presented in \cite{ACN08} with no compromise in the approximation guarantees.  Mathieu and Schudy in \cite{MS10clustering}
considered the \emph{planted graph} version, in which the input is a  noisy version of a  union-of-cliques graph, and show that
a PTAS is possible for this setting.
Also, Giotis et. al \cite{Giotis:2006:CCF:1109557.1109686} and independently using other techniques, Karpinski et. al \cite{KS08} gave a PTAS for the  \CC case in which the number of clusters is constant.

Amit \cite{NogaAmit04} was the first to address \BCC directly. She proved its NP-hardness and gave a constant $11$-approximation algorithm based
on rounding a linear programming in the spirit of Charikar et. al's \cite{CharikarGW05} algorithm for \CC.

It is worth noting that in \cite{GHKZ08} a $4$-approximation algorithm for \BCC was presented and analyzed.
The presented algorithm is incorrect (we give a counter example in the paper) but their attempt
to use arguments from \cite{ACN08} is an excellent one. We will show that an extension of the method
in \cite{ACN08} is needed.

\subsection{Our Results}

Our main result,  requiring a considerable development of previous techniques, is a randomized expected $4$-approximation algorithm, \QBCb.

To explain how we attain it, we recall the method of Ailon et. al \cite{ACN08}.  The algorithm for \CC presented there is as follows
(we concentrate on the unweighted case).  Choose a random vertex, and form a cluster with its neighbors.  Remove the cluster from the graph,
and repeat until the graph is empty.  This  random-greedy algorithm returns a solution with cost at most $3$-times that of the optimal solution, on expectation.
The analysis was done by noticing that each cost element is naturally related to a \emph{contradiction structure} containing
$3$ vertices and exactly $2$ edges between them.  This structure is, incidentally, the minimal structure  forcing any solution to pay.
In other words, the locations in which any clustering errs must \emph{hit} the set of contradicting structures.
A corresponding hitting set LP lower bounding the optimal solution was defined to capture this simple observation, and a feasible solution
was then conveniently assigned to its dual using probabilities arising in the algorithm probability space.

It is tempting here to consider the corresponding  minimal contradiction structure for \BCC, namely a set of $4$ vertices, $2$ on each side,
with exactly $3$ edges between them.  Unfortunately, this idea turned out to be evasive (a proposed solution attempting this \cite{GHKZ08} has a counter example
which we describe and analyze in Appendix~\ref{section:counter} and
is hence incorrect).  In our analysis we resorted to
contradiction structures of unbounded size.  Such a structure consists of two vertices $\ell_1, \ell_2$ of the left side and
two sets  of vertices $N_1,N_2$ on the right hand side
such that $N_i$ is contained in the neighborhood of $\ell_i$ for $i=1,2$, $N_1\cap N_2 \neq \emptyset$ and $N_1 \neq N_2$.
We define a hitting LP as we did earlier, this time of possibly exponential size, and analyze its dual in tandem with a carefully
constructed random-greedy algorithm.  As this analysis sketch suggests, the algorithm is not symmetrical with respect to the right and left side
of the input.  Indeed, at each round it chooses a random pivot vertex on the left, constructs a cluster with its right hand side neighbors, and then for each other vertex
on the left hand side makes a randomized decision whether to join the new cluster based on the intersection pattern of  its neighborhood  with the pivot's neighborhood.

\subsection{Paper Structure}
We start with basic notation in Section~\ref{sec:notation}.
We then present our main algorithm in Section~\ref{sec: the algorithm}, followed by its analysis in Section~\ref{sec:analysis}.
We discuss future work in Section~\ref{sec:future}.

\section{Notation}\label{sec:notation}
Before describing the framework we give some general facts and notations.
Let the input graph be $G=(L,R,E)$ where $L$ and $R$ are the sets of left and right nodes and $E$ be
a subset of $L\times R$. Each element $(\ell,r) \in L\times R$ will be referred to as a \emph{pair}.

A solution to our combinatorial problem is a clustering $C_1, C_2, \dots, C_m$ of the set $L\cup R$.  We  identify such a clustering with a bipartite graph
$B=(L,R,E_{B})$  for which $(\ell,r) \in E_B$ if and only if $\ell \in L$ and $r\in R$ are  in the same cluster $C_i$ for some $i$.  Note that given $B$, we
are unable to identify clusters contained exclusively in  $L$ (or $R$), but this will not affect the cost, so we adopt the convention that
single-side clusters are always singletons.

We will say that a pair $e = (\ell,r)$ is erroneous if $e \in (E \setminus E_B)\cup(E_B \setminus E)$.
For convenience, let $x_{G,B}$ be the indicator function for the erroneous pair set, i.e., $x_{G,B}(e) =1$ if $e$ is erroneous and $0$ otherwise.
We will also simply use $x(e)$ when it is obvious to which graph $G$ and clustering $B$ it refers.
The cost of a clustering solution is defined to be $\cost_{G}(B) = \sum_{e \in L\times R} x_{G,B}(e)$.
Similarly, we will use $\cost(B) = \sum_{e \in L\times R} x(e)$ when $G$ is clear from the context,
Let $N(\ell) = \{r | (\ell,r) \in E\}$ be the set of all right nodes adjacent to $\ell$.

It will be convenient for what follows to define a \emph{tuple}.
We define a tuple $T$ to be $T= (\ell^T_1,\ell^T_2,R^T_{1},R^T_{1,2},R^T_2)$
where $\ell^T_1, \ell^T_2 \in L$, $\ell^T_1 \neq \ell^T_2$, $R^T_{1} \subseteq N(\ell^T_1)\setminus N(\ell^T_2)$,  $R^T_{2} \subseteq N(\ell^T_2)\setminus N(\ell^T_1)$
and $R^T_{1,2} \subseteq N(\ell^T_2)\cap N(\ell^T_1)$.  In what follows, we may omit the superscript of $T$.
Given a tuple $T=(\ell^T_1,\ell^T_2,R^T_{1},R^T_{1,2},R^T_2)$, we define the \emph{conjugate tuple}
$\bar{T}  =  (\ell^{\bar{T}}_1,\ell^{\bar{T}}_2,R^{\bar{T}}_{1},R^{\bar{T}}_{1,2},R^{\bar{T}}_2) = (\ell^T_2,\ell^T_1,R^T_{2},R^T_{1,2},R^T_1)$.
Note that $\bar{\bar T} = T$.

\section{The Algorithm}\label{sec: the algorithm}

We now describe our algorithm \QBCb.  The algorithm is
sequential. In every cycle it creates one cluster and possibly
many singletons, all of which are removed from the graph before continuing to the next iteration.  
Abusing notation, by $N(\ell)$ we mean, in the algorithm's description, all the neighbors of $\ell\in L$ which have not yet
been removed from the graph.

Every such cycle performs two phases.
In the first phase, \QBCb picks a node on the left side
uniformly at random, $\ell_1$, and forms a new cluster $C=\{\ell_1\} \cup N(\ell_1)$.
This will be referred to as the $\ell_1$-phase and $\ell_1$ will be referred to as the left center of the cluster.
In the second phase, denoted as the  $\ell_2$-sub-phase corresponding to the $\ell_1$-phase, the algorithm iterates over all other remaining left nodes, $\ell_2$,
and decides either to (1) append them to  $C$, (2) turn them into singletons, or (3) do nothing.
We now explain how to make this decision. let $R_1= N(\ell_1) \setminus  N(\ell_2)$, $R_2= N(\ell_2) \setminus  N(\ell_1)$ and
$R_{1,2}= N(\ell_1)\cap N(\ell_2)$.
With probability $\min\{\frac{|R_{1,2}|}{|R_{2}|} , 1\}$ do one of two things:
(1) If $|R_{1,2}| \ge |R_1|$ append $\ell_2$ to $C$, and otherwise
(2) (if $|R_{1,2}| < |R_1|$), turn $\ell_2$ into a singleton.
In the remaining probability, (3) do nothing for $\ell_2$, leaving it in the graph for future iterations.
Examples for cases the algorithm encounters for different ratios of $R_1$, $R_{1,2}$, and $R_2$ are given in Figure~\ref{fig:algorithm}.

\begin{figure}%
\label{fig:algorithm}
\begin{center}%
    \includegraphics[width=17cm]{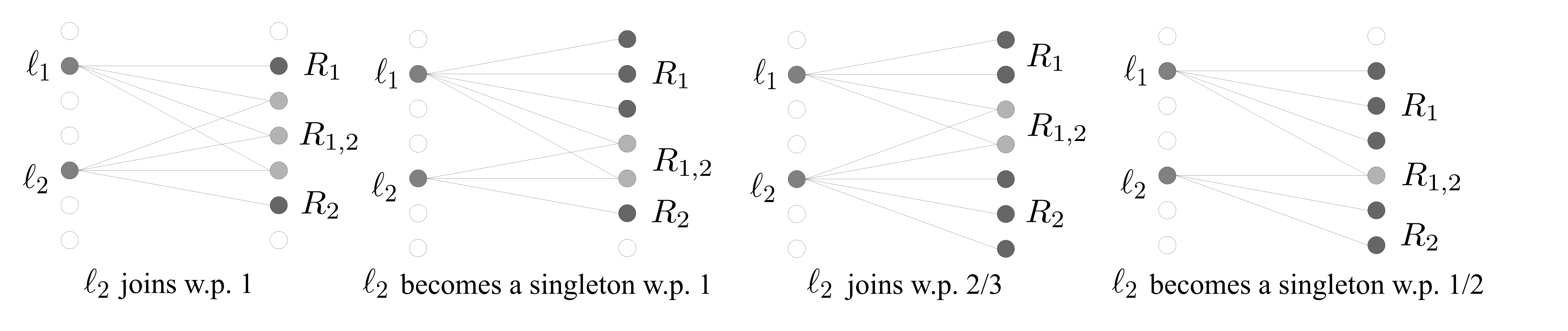}
\end{center}%
\caption{Four example cases in which $\ell_2$ either joins the cluster created by $\ell_1$ or becomes a singleton.
In the two right most examples, with the remaining probability nothing is decided about $\ell_2$.}
\end{figure}%

\begin{theorem}\label{thm:main}
 Algorithm \QBCb returns a solution with expected cost at most $4$ that of the optimal solution.
\end{theorem}

\section{Algorithm Analysis}\label{sec:analysis}
We start by describing {\em bad events}.
This will help us relate the expected cost of the algorithm to a sum of event probabilities and expected consequent costs.
\begin{definition}
We say that a {\em bad event}, $X_T$, happens to the tuple
$T = (\ell_1^T, \ell_2^T, R_1^T, R_{1,2}^T, R_{2}^T)$ if during the execution of \QBCb, $\ell^T_1$ was chosen
to be a left center while $\ell^T_2$
was still in the graph, and \emph{at that moment}, $R^T_1= N(\ell^T_1)\setminus N(\ell^T_2)$,
$R^T_{1,2}= N(\ell^T_1)\cap N(\ell^T_2)$, and
$R^T_{2}= N(\ell^T_2)\setminus N(\ell^T_1)$.  (We refer by $N(\cdot)$ here to the neighborhood function in a particular
moment of the algorithm execution.)
\end{definition}
If a bad event $X_T$ happens
to tuple $T$ we ``color'' the following pairs with color $T$ :
\begin{itemize}
\item   $\{(\ell^T_2,r_{1})\ : \ r_{1}\in R^T_{1}\}$,
\item   $\{(\ell^T_2,r_{1,2})\ :\  r_{1,2}\in R^T_{1,2}\}$,
\item  $\{(\ell^T_2,r_2)\ :\  r_2\in R^T_2\}$ only if we decide
        to associate $\ell^T_2$ to $\ell^T_1$'s cluster, or if we decide to make
        $\ell^T_2$ a singleton during the $\ell_2$-sub-phase corresponding to the $\ell_1$-phase.
\end{itemize}

\begin{lemma}\label{lm: mapping errors to tuples}
During the execution of \QBCb each pair $(\ell,r) \in L\times R$
is colored at most once, and each pair on which the output errs is colored exactly once.
\end{lemma}
\begin{proof}
For the first part, we  show that pairs are colored at most once.
A pair $(\ell,r)$ can only be colored during an $\ell_2$-sub-phases with respect to some $\ell_1$-phase, if $\ell=\ell_2$.
Clearly, this will only happen in one $\ell_1$-phase, as
every time a pair is colored either $\ell_2$ or $r$ (or both) are removed from the graph.
Indeed, either $r \in R_1 \cup R_{1,2}$ in which case $r$ is removed, or $r \in R_2$,  but  then $\ell$ is
removed since it either joins the cluster created by $\ell_1$ or becomes a singleton.

For the second part, note that the only pairs which are not colored are between left centers
(during $\ell_1$-phases) and right nodes in the graph at that time.
On all these  pairs the algorithm does not err.
\end{proof}

We denote by $q_T$ the probability that event $X_T$ occurs and by
$\cost(T)$ the number of erroneous pairs
that are colored by $X_T$.
From Lemma~\ref{lm: mapping errors to tuples}
we get the following:

\begin{corollary}\label{obs: mapping QBCb to tuples}

$$
\mathbb{E}[\cost[\QBCb]] = \mathbb{E}\left[\sum_{e \in L\times R}  x(e) \right] = \mathbb{E}\left[\sum_T \cost(T) \right]
=\sum_T {q_T \cdot \mathbb{E}[\cost(T) | X_T]}\ .
$$
\end{corollary}

Note: In what follows we use the terms \emph{erroneous pairs} and \emph{violating pairs} or \emph{violation pairs} interchangingly, referring to pairs on which the algorithm incurs a unit of cost.
\subsection{Contradicting Structures}
We now identify bad structures in the graph for
which every output must incur some cost.
In the case of \BCC the minimal such structures are ``bad squares'': 	
A set of four nodes, two on each side, between which there are only three edges.
We make the trivial observation that any clustering $B$ must make at least one mistake on any
such bad square, $s$ (we think of $s$ as the set of $4$ pairs connecting its two left nodes and two right nodes).
Any clustering solution's violating pair set must  hit these squares.
Let $S$ denote the set of all bad squares in the input graph $G$.

It is not enough to concentrate on squares in our analysis.  Indeed, at an $\ell_2$-sub-phase, decisions are made
based on the intersection pattern of the current neighborhoods of $\ell_2$ and $\ell_1$ - a possibly unbounded structure.
The \emph{tuples} now come in handy.

Consider tuple  $T=(\ell^T_1,\ell^T_2,R^T_{1},R^T_{1,2},R^T_2)$
for which $|R^{T}_{1,2}| > 0 $ and $|R^{T}_{2}  | > 0$ .
Notice that for every selection of $r_2\in R^T_2$,
and $r_{1,2}\in R^T_{1,2}$ the tuple contains the
bad square induced by $\{\ell_1, r_2,\ell_2 , r_{1,2}\}$.  Note that there may
also be bad squares $\{\ell_2, r_1,\ell_1,r_{1,2}\}$ for every
$r_1\in R^T_1$ and $r_{1,2}\in R^T_{1,2}$ but these will be
associated to the \emph{conjugate tuple}
$\bar{T}=(\ell^T_2,\ell^T_1,R^T_{2},R^T_{1,2},R^T_1)$.

For each tuple we can write a corresponding linear constraint on the function $\{x(e): e \in L\times R\}$, indicating, as we explained above, the pairs for which
the algorithm errs.
A tuple constraint is the sum of the constraints
of the squares it is associated with, where a constraint for square $s$ is simply defined as $\sum_{e\in s} x(e) \geq 1$.
Since each tuple corresponds to $|R^T_{2}|\cdot |R^T_{1,2}|$
bad squares, we get the following constraint:
\begin{eqnarray*}
\lefteqn{\forall\ T:\ \  \sum_{r_2\in R^T_{2},r_{1,2}\in R^T_{1,2}}
\left(x_{\ell^T_1,r_2} + x_{\ell^T_1,r_{1,2}} +
x_{\ell^T_2,r_2} + x_{\ell^T_2,r_{1,2}} \right) =} \\
&&\sum_{r_2\in R^T_{2}}{|R^T_{1,2}|\cdot(x_{\ell^T_1,r_2} + x_{\ell^T_2,r_2})}
+ \sum_{r_{1,2}\in R^T_{1,2}}{|R^T_{2}|\cdot(x_{\ell^T_1,r_{1,2}}
+ x_{\ell^T_2,r_{1,2}})} \geq |R^T_{2}|\cdot |R^T_{1,2}|
\end{eqnarray*}
The following linear program hence provides a lower bound for the optimal solution:

\begin{eqnarray*}
LP &=& min \sum_{e \in L \times R} x(e) \\
&\mbox{s.t.}\;\;\forall \, T&\frac{1}{|R^T_{2}|}\sum_{r_2\in R^T_{2}}(x_{\ell^T_1,r_2}
+ x_{\ell^T_2,r_2}) +  \frac{1}{|R^T_{1,2}|}\sum_{r_{1,2}\in R^T_{1,2}}
(x_{\ell^T_1,r_{1,2}} + x_{\ell^T_2,r_{1,2}}) \geq 1
\end{eqnarray*}

Notice that all the constraints in this program are sums of square constraints.  This means that the program is equivalent to one in which
only square constraints are present.  Our formulation, however, allows the definition of useful dual variables corresponding to each tuple $T$.
The dual program is as follows:

\begin{eqnarray*}
DP  &=& max \sum_{T} \beta(T) \\
&\mbox{s.t.} \,\, \forall \, (\ell,r) \in E: &\sum_{T :\, \ell_2^T = \ell, r \in R_{2}^T}{\frac{1}{|R_{2}^T|}\beta(T)} +
\hspace{-.5cm}\sum_{T:\, \ell_1^T=\ell,\, r \in R_{1,2}^T}{ \frac{1}{|R_{1,2}^T|}\beta(T)}+
\hspace{-.5cm}\sum_{T:\,\ell_2^T=\ell,\, r \in R_{1,2}^T}{ \frac{1}{|R_{1,2}^T|}\beta(T)} \leq 1\\
&\mbox{and } \forall \,(\ell,r) \not\in E: &\sum_{T:\, \ell_1^T=\ell,\, r \in R_{2}^T}
{\frac{1}{|R_{2}^T|}\beta(T)} \leq 1
\end{eqnarray*}

\subsection{Obtaining the Competitive Analysis}\label{sec: competitive analysis}

We now
relate the
expected cost of the algorithm on each tuple
to a feasible solution for $DP$.  We remind the reader that $q_T$ denotes the probability that a bad event $X_T$ happens to tuple $T$.
\begin{lemma}
Let $\beta(T) = \alpha_T\cdot q_T \cdot \min \{|R^T_{1,2}|,|R^T_{2}|\}$ , when \\
\[
\alpha_T=\min\left\{1 , \frac{|R^T_{1,2}|}{\min\{|R^T_{1,2}|,|R^T_{1}|\}+\min\{|R^T_{1,2}|,|R^T_{2}|\}}\right\}
\]
then $\beta$ is a feasible solution to $DP$. 

In other words,
for every edge $e=(\ell,r)\in E$:
\begin{equation}\label{eq: e in E}
\sum_{T\  s.t\ \ell_2^T = \ell, r \in R_{2}^T}{\frac{1}{|R_{2}^T|} \beta(T)} +
\sum_{T\  s.t\  \ell_1^T=\ell,\, r \in R_{1,2}^T}{ \frac{1}{|R_{1,2}^T|} \beta(T)}+
\sum_{T\  s.t\ \ell_2^T=\ell,\, r \in R_{1,2}^T}{ \frac{1}{|R_{1,2}^T|} \beta(T)} \leq 1.
\end{equation}
And for every pair $e= (\ell,r) \not\in E$:
\begin{equation}\label{eq: e not in E}
\sum_{T\ s.t\ \ell_1^T=\ell,\, r \in R_{2}^T}{\frac{1}{|R_{2}^T|} \beta(T)} \leq 1\ .
\end{equation}
\end{lemma}

\begin{proof}

First, notice that given a pair $e=(\ell,r)\in E$ each
tuple $T$ can appear at most in one of the sums
in the LHS of (\ref{eq: e in E}). 
Denote by $X_{e,T}$ the event that the edge $e$ is colored with color $T$.  
We distinguish between two cases.
\begin{enumerate}
\item Consider $T$ appearing in the first sum of the LHS of (\ref{eq: e in E}),
    meaning that $\ell^T_2=\ell$ and  $r\in R^T_2$. We distinguish between two sub-cases.
    \begin{itemize}
    \item If  $|R^T_{1,2}|\geq |R^T_1|$, $e$ is colored with color $T$
        if $\ell^T_2$ joined the cluster of $\ell^T_1$.
        This happens, conditioned on $X_T$, with probability $\Pr[X_{e,T}|X_T]=
        \min \left\{\frac{ |R^T_{1,2}| }{|R^T_{2}|} , 1\right\}$,
    \item if $|R^T_{1,2}|< |R^T_1|$ we color $e$ with color $T$
        if $\ell_2$ was isolated, which happens with probability
        $\Pr[X_{e,T}|X_T]=\min \{\frac{ |R_{1,2}| }{|R_{2}|} , 1\}$ as well.
    \end{itemize}
    Thus, $T$ contributes the following expression to the sum:
    \begin{eqnarray*}
    {\frac{1}{|R^T_{2}|}\beta(T)} = \frac{1}{|R^T_{2}|}
         \alpha_T\cdot q_T\cdot \min \{|R^T_{1,2}|,|R^T_{2}|\} &\leq&
    q_T\cdot \min \left \{\frac {|R^T_{1,2}|} {|R^T_{2}|} , 1\right \}\\ &=&
    \Pr [X_T] \Pr[X_{e,T}|X_T]  =  \Pr[X_{e,T}].
    \end{eqnarray*}
\item
 $T$ contributes to the second or third sum in the LHS of (\ref{eq: e in E}).
    By definition of the conjugate $\bar T$, the following holds:
    \begin{equation}\label{monkey}
    \sum_{T\  s.t\  \ell_1^T=\ell,\, r \in R_{1,2}^T}{ \frac{1}{|R_{1,2}^T|} \beta(T)}+
\hspace{-.4cm}\sum_{T\  s.t\ \ell_2^T=\ell,\, r \in R_{1,2}^T}{ \frac{1}{|R_{1,2}^T|} \beta(T)}=
\hspace{-.4cm}\sum_{T\  s.t\  \ell_1^T=\ell,\, r \in R_{1,2}^T}{ \frac{1}{|R_{1,2}^T|} \left(\beta(T)+\beta(\bar{T})\right)}.
    \end{equation}
    Therefore it is sufficient to bound the contribution of each $T$
    to the  RHS of (\ref{monkey}).  We may therefore focus on tuples $T$ for which
    if $\ell=\ell_1^T$ and $r \in R_{1,2}^T$.
    Consider a moment in the algorithm's execution in which both $\ell_1^T$ and $\ell_2^T$
    were still present in the graph, $R^T_1= N(\ell^T_1)\setminus N(\ell^T_2)$,
    $R^T_{1,2}= N(\ell^T_1)\cap N(\ell^T_2)$,
    $R^T_{2}= N(\ell^T_2)\setminus N(\ell^T_1)$
    and one of  $\ell^T_1, \ell^T_2$ was chosen to be a left center.\footnote{We use the definition of $N(\cdot)$ which depends on the ``current'' state of the graph at that moment, after possibly removing previously created clusters.}
    Either one of $\ell_1^T$ and $\ell_2^T$ had the same
    probability to be chosen.  In other words:
$$ \Pr[X_T | X_T \cup X_{\bar T}] = \Pr[X_{\bar T} | X_T \cup X_{\bar T}]\ ,$$
and hence, $q_T = q_{\bar T}$.
    Further, notice that $e=(\ell,r)$ is never colored with color $T$,
    and if event $X_{\bar T}$ happens then $e$ is colored with color $\bar T$
    with probability 1.
    Therefore:
    \begin{eqnarray*}
   \lefteqn{\frac{1}{|R_{1,2}^T|} \left(\beta(T)+ \beta(\bar{T})\right)  }\\
    &=&  \frac{1}{|R_{1,2}^T|}\cdot q_T \cdot \min\left\{1 , \frac{|R^T_{1,2}|}{\min\{|R^T_{1,2}|,|R^T_{1}|\}+\min\{|R^T_{1,2}|,|R^T_{2}|\}}\right\}\\
    &&\hspace{7cm} \cdot\left(\min \{|R^T_{1,2}|,|R^T_{2}|\}+ \min \{|R^{\bT}_{1,2}|,|R^{\bT}_{2}|\}\right)\\
    &\leq& q_T = q_{\bar T}= \Pr[X_{\bar T}] = \Pr[X_{e, \bar T}]+\Pr[X_{e, T}] .
    \end{eqnarray*}
\end{enumerate}
Summing this all together, for every edge $e\in E$:
$$\sum_{T\  s.t\ \ell_2^T = \ell, r \in R_{2}^T}
{\frac{1}{|R_{2}^T|} \beta(T)} +
\sum_{T\  s.t\  \ell_1^T=\ell,\, r \in R_{1,2}^T}
{ \frac{1}{|R_{1,2}^T|} \beta(T)}+
\sum_{T\  s.t\ \ell_2^T=\ell,\, r \in R_{1,2}^T}
{ \frac{1}{|R_{1,2}^T|} \beta(T)}
\leq \sum_T{ Pr[X_{e,T}]}.
$$
By the first part of Lemma~\ref{lm: mapping errors to tuples} we know
that $\sum_T{ Pr[X_{e,T}]}$ is exactly
the probability of the edge $e$ to be colored (the sum is over probabilities of disjoint events),
therefore it is at most $1$, as required to satisfy~(\ref{eq: e in E}).

Now consider  a pair $e=(\ell,r) \not\in E $.  A tuple $T$ contributes
to (\ref{eq: e not in E}) if
$\ell^T_1=\ell$ and $r\in R^T_{2}$.
Since, as before, $q_T=q_{\bT}$ and since $\Pr[X_{e,\bar{T}}|X_{\bar T}] =1$
(this follows from the first coloring rule described in the beginning of
Section~\ref{sec:analysis}) we obtain the following:
\begin{eqnarray*}
\sum_{T\ s.t\ \ell_1^T=\ell,\, r \in R_{2}^T}{\frac{1}{|R_{2}^T|} \beta(T)}
&=& \sum_{T\ s.t\ \ell_1^T=\ell,\, r \in R_{2}^T}
{\frac{1}{|R_{2}^T|}\cdot \alpha_T \cdot q_T \cdot \min \{|R^T_{1,2}|,|R^T_{2}|\}}\\
 &\le & \sum_{T\ s.t\ \ell_1^T=\ell,\, r \in R_{2}^T}{ q_T } \ = \
    \sum_{\bar{T}\ s.t\ \ell_2^{\bar{T}}=\ell,\, r \in R_{1}^{\bar{T}}}{ q_{\bar{T}} }\\
 &= & \sum_{\bar{T}\ s.t\ \ell_2^{\bar{T}}=\ell,\, r \in R_{1}^{\bar{T}}}  { \Pr [X_{\bT}] }
 =      \sum_{\bar{T}\ s.t\ \ell_2^{\bar{T}}=\ell,\, r \in R_{1}^{\bar{T}}}  { \Pr [X_{e,\bT}] }  \\  
 &=& \sum_{T}{ \Pr[X_{e,T}]}.
\end{eqnarray*}
From the same reason as before, this is at most $1$, as required for (\ref{eq: e not in E}).
\end{proof}

After presenting the feasible solution to our
dual program, we have left to prove that the expected
cost of \QBCb is at most 4 times the  DP value of this solution.
For this we need the following:

\begin{lemma}\label{lm: tuple cost}
For any tuple $T$,
$$ q_T\cdot\E[\cost(T)|X_T] + q_{\bar T} \cdot \E[\cost({\bar{T}})|X_{\bar T}] \leq
4\cdot\left(\beta(T) +\beta(\bar{T})\right) .$$
\end{lemma}
\begin{proof}
We consider three cases, according to the structure of $T$.

\noindent
{\bf Case 1.} $|R^T_1|\leq |R^T_{1,2}|,\ |R^T_{2}|\leq |R^T_{1,2}|$ (equivalently
            $|R^{\bar{T}}_1|\le |R^{\bar{T}}_{1,2}| ,\ |R^{\bar{T}}_2|\le |R^{\bar{T}}_{1,2}|$) :\\
        For this case, $\alpha_T=\alpha_{\bT}=\min\left\{1,\frac{|R^T_{1,2}|}{|R^T_{1}|+|R^T_2|}\right\} $,
         and we have (recall that $q_T = q_{\bar T}$)
         \begin{eqnarray*}
        \beta(T) +\beta(\bar{T}) &=&  \alpha_T \cdot q_T\cdot \left(\min \{|R^T_{1,2}|,|R^T_{2}|\}
        + \min \{|R^T_{1,2}|,|R^T_1|\}  \right)\\
        &=& q_T\cdot \min\{(|R^T_{2}|+|R^T_1|),|R^T_{1,2}|\}
        \ge \frac{1}{2}  \cdot q_T\cdot (|R^T_{2}|+|R^T_1|).
         \end{eqnarray*}
        Since $|R^T_1|\leq |R^T_{1,2}|$, if event $X_T$ happens \QBCb
        adds $\ell^T_2$ to $\ell^T_1$'s cluster with probability
        $\min \left\{\frac {|R^T_{1,2}|}{|R^T_{2}|} , 1\right\}= 1$. Therefore
        the pairs colored with color $T$ that \QBCb violates  are
        all the edges from $\ell^T_2$ to $R^T_{2}$ and
        all the non-edges from $\ell^T_2$ to $R^T_1$, namely,
        $|R^T_{2}|+|R^T_1|$ edges. The same happens in the event $X_{\bar{T}}$
        as the conditions on $|R^{\bar{T}}_1|$, $|R^{\bar{T}}_{1,2}|$, and
        $|R^{\bar{T}}_2|$ are the same,
        and since $|R^{\bar{T}}_{2}|+|R^{\bar{T}}_1|= |R^T_1|+|R^T_{2}|$.
        Thus,
        $$
        q_T \cdot \left( \E[\cost(T|X_T)]+ \E[\cost({\bar{T}}|X_{\bar T})]\right)
        = q_T \left( 2\left( |R^T_{2}|+|R^T_1|\right) \right)
        \le 4\cdot \left(\beta(T)+\beta({\bar{T}}) \right).
        $$

\noindent
{\bf Case 2.}  $|R^T_1|< |R^T_{1,2}| < |R^T_{2}|$  (equivalently
            $|R^{\bar{T}}_1|> |R^{\bar{T}}_{1,2}| > |R^{\bar{T}}_{2}|$) :\\
            Here $\alpha_T=\alpha_{\bT}=\min\left\{1,\frac{|R^T_{1,2}|}{|R^T_{1}|+|R^T_{1,2}|}\right\} $,
            therefore,
        \begin{eqnarray*}
        \beta(T) +\beta(\bar{T}) &=& \alpha_T\cdot q_T \cdot
        \left(\min \{|R^T_{1,2}|,|R^T_{2}|\} +
        \min \{|R^T_{1,2}|,|R^T_1|\}  \right)\\
       &=&  q_T \cdot \min\{ |R^T_{1,2}|+|R^T_1|, |R^T_{1,2}|\}
       =  q_T \cdot |R^T_{1,2}|.
       \end{eqnarray*}
        As $|R^T_1|\leq |R^T_{1,2}|$, if event $X_T$ happens \QBCb
        adds $\ell^T_2$ to $\ell^T_1$ cluster with probability
        $\min \left\{\frac {|R^T_{1,2}|}{|R^T_{2}|} , 1\right\}
        = \frac {|R^T_{1,2}|}{|R^T_{2}|}$.
        Therefore with probability $\frac {|R^T_{1,2}|}{|R^T_{2}|}$
        the pairs colored by color  $T$ that \QBCb violate are all
        the edges from $\ell^T_2$ to $R^T_{2}$ and all the
        non-edges from $\ell^T_2$ to $R^T_1$, and with
        probability $\left(1-\frac {|R_{1,2}|}{|R_{2}|}\right)$
        \QBCb violates all the edges from $\ell^T_2$ to $R^T_{1,2}$.
        Thus,
        \begin{eqnarray*}
        \E[\cost(T)|X_T]&=&\frac {|R^T_{1,2}|}{|R^T_{2}|} \left(|R^T_{2}|+|R^T_1|\right) +
                      \left(1-\frac {|R^T_{1,2}|}{|R^T_{2}|}\right)|R^T_{1,2}|\\
                 &=& 2\cdot |R^T_{1,2}| + \frac {|R^T_{1,2}|\cdot |R^T_1|-|R^T_{1,2}|^2}
                 {|R^T_{2}|}\ \leq\  2\cdot|R^T_{1,2}|.
        \end{eqnarray*}
        If the event $X_{\bar{T}}$ happens, as
        $|R^{\bar{T}}_1|>|R^{\bar{T}}_{1,2}|$
        and $\min \left\{\frac {R^{\bar{T}}_{1,2}}{ R^{\bar{T}}_2} , 1\right\}= 1$,
        \QBCb chooses to isolate $\ell^{\bar{T}}_2$ ($=\ell^T_1$)
        with probability $1$ and the number of pairs colored with color ${\bar{T}}$ that are consequently
        violated are  $|R^{\bar{T}}_2|+|R^{\bar{T}}_{1,2}|=|R^{T}_1|+|R^{T}_{1,2}|$ .
        Thus,
        \begin{eqnarray*}
        q_T\cdot\left(\E[\cost(T)|X_T])+ \E[\cost(\bar T) | X_{\bar{T}})]\right)
        &\leq& q_T \cdot ( 2|R^{T}_{1,2}|+ |R^{T}_1|+|R^{T}_{1,2}|)\\
        &<& 4 \cdot q_T \cdot |R^T_{1,2}|
        = 4 \cdot \left(\beta(T)+\beta(\bar{T}) \right)\ .
        \end{eqnarray*}

\noindent
{\bf Case 3.} $|R^T_{1,2}|< |R^T_1|, |R^T_{1,2}|<|R^T_{2}|$ (equivalently,
        $|R^{\bar{T}}_{1,2}|< |R^{\bar{T}}_2|, |R^{\bar{T}}_{1,2}|<|R^{\bar{T}}_{1}|$):\\
        Here, $\alpha_T=\alpha_{\bT} = \frac 1 2$ , thus,
        $$
        \beta(T) +\beta(\bar{T})= \frac {1}{2}\cdot q_T
            \cdot \left(\min \{|R^T_{1,2}|,|R^T_{2}|\}
                + \min \{|R^T_{1,2}|,|R^T_1|\}  \right)
            = q_T\cdot |R^T_{1,2}|\ .
        $$
        Conditioned on event $X_T$, as $|R^T_1| > |R^T_{1,2}|$, \QBCb chooses
        to isolate $\ell_2$ with probability
        $\min \left\{\frac {|R^T_{1,2}|}{|R^T_{2}|} , 1\right\}
        = \frac {|R^T_{1,2}|}{|R^T_{2}|}$.
        Therefore with probability $\frac {|R^T_{1,2}|}{|R^T_{2}|}$ \QBCb
        colors $|R^T_{2}|+|R^T_{1,2}|$ pairs with color $T$ (and violated them all).
        With probability $\left(1-\frac {|R^T_{1,2}|}{|R^T_{2}|}\right)$, \QBCb
        colors $|R^T_{1,2}|$ pairs with color $T$ (and violated them all). We conclude that
        $$
        \E[\cost(T)|X_t] =\frac {|R^T_{1,2}|}{|R^T_{2}|} (|R^T_{2}|+|R^T_{1,2}|)+
        \left(1-\frac {|R^T_{1,2}|}{|R^T_{2}|}\right) |R^T_{1,2}| = 2 |R^T_{1,2}|\ .
        $$
        Similarly, for event $X_{\bar{T}}$, as $|R^{\bar{T}}_1|>|R^{\bar{T}}_{1,2}|$
        and  $\min \left\{\frac {|R^{\bar{T}}_{1,2}|} {|R^{\bar{T}}_2|} , 1\right\}
        = \frac {|R^T_{1,2}|} {|R^T_1|}$,
        \QBCb isolates $\ell_1$ with probability
        $\frac {|R^T_{1,2}|} {|R^T_1|}$
        therefore colors $|R^{\bT}_2|+|R^{\bT}_{1,2}|$
         pairs  with color $\bar{T}$ (and violated them all).
        With probability $(1-\frac {|R^T_{1,2}|} {|R^T_1|})$
        \QBCb colors $|R^{\bT}_{1,2}|$ pairs with color ${\bT}$ (and violates them all).
        Thus,
        $$
        \E[\cost(\bar{T})|X_{\bar T}]=\frac {|R^T_{1,2}|}{|R^T_1|} (|R^T_1|+|R^T_{1,2}|)
        + \left(1-\frac {|R^T_{1,2}|}{|R^T_1|}\right)|R^T_{1,2}| = 2 |R^T_{1,2}| .
        $$
        And therefore
        $$q_T\cdot\left(\E[\cost(T)|X_t]+ \E[\cost(\bar T) | X_{\bar{T}}]\right)
         = 4\cdot q_T \cdot |R^T_{1,2}|
        = 4\cdot (\beta(T) +\beta(\bar{T})) \ .$$

\end{proof}

\noindent
By Corollary~\ref{obs: mapping QBCb to tuples}
\begin{eqnarray*}
 E[\QBCb] &=& \sum_T {\Pr[X_T]\cdot \E[\cost(T)|X_T]} \\ &=&\frac{1}{2} \sum_{T}{\left(\Pr[X_T] \cdot \E[\cost(T)|X_T]+\Pr[X_{\bT}]\cdot \E[\cost(\bar{T})|X_{\bar T}]\right)}\ .
\end{eqnarray*}
By Lemma~\ref{lm: tuple cost} the above RHS is at most
$2\cdot \sum_{T} ({\beta(T) +\beta(\bar{T})}) = 4\cdot\sum_{T} {\beta(T)}.$
Therefore by the weak duality theorem we conclude that
$$\mathbb{E}[\QBCb]\le 4\cdot\sum_{T} {\beta(T)}\leq 4 \cdot OPT .$$
This proves our main result Theorem~\ref{thm:main}.

\section{Future Work}\label{sec:future}

Improving the approximation factor as well as derandomizing the algorithm (in the lines of \cite{ZHJW07}, or using other techniques) are interesting questions.
One direction that seems promising is to devise an LP rounding algorithm using a variation of \QBCb (in the lines of the LP-based algorithms in \cite{ACN08}).

\appendix
\section{A Counter Example for a Previously Claimed Result}\label{section:counter}
In \cite{GHKZ08} the authors claim to design and analyze a $4$-approximation algorithm for \BCC.  Its analysis is based on
 bad squares (and not unbounded structures, as done in our analysis).
Their algorithm is as follows: First, choose a pivot node uniformly at randomly from the left side, and cluster it with all its neighbors. Then, for each node on the left,
if it has a neighbor in the newly created cluster, append it with probability $1/2$. An exception is reserved for nodes whose
neighbor list is identical  that of the pivot, in which case these nodes join with probability $1$. Remove the clustered nodes and repeat until
no nodes are left in the graph.

Unfortunately,  there is an example demonstrating that the algorithm has an unbounded approximation ratio.
Consider a bipartite graph on $2n$ nodes, $\ell_{1,\ldots,n}$ on the left and $r_{1,\ldots,n}$ on the right.
Let each node $\ell_{i}$ on the left be connected to all other nodes on the right except for $r_{i}$.
The optimal clustering of this graph connects all $\ell_i$ and $r_i$ nodes and thus has cost $OPT=n$.
In the above algorithm, however, the first cluster created will include all but one of the nodes on the
right and roughly half the left ones. This already incurs a cost of $\Omega(n^2)$ which is a factor $n$ worse than the best possible.

\noindent
As a side note, the authors of this abstract have also tried to design an algorithm based on an analysis involving squares only, to no avail.

\end{spacing}
\end{document}